\documentclass[a4pape,USenglish]{lipics-v2021}
\usepackage{algorithm}
\usepackage{amsmath}
\usepackage{algpseudocode}
\usepackage{tikz}
\usepackage[most]{tcolorbox}
\usepackage[colorinlistoftodos,textsize=tiny]{todonotes}

\newcommand{\adj}{\mathsf{adj}}
\newcommand{\inc}{\mathsf{inc}}
\newcommand{\red}{\mathsf{red}}
\newcommand{\blue}{\mathsf{blue}}
\newcommand{\id}{\mathsf{id}}
\newcommand{\td}{\mathsf{td}}
\newcommand{\tw}{\mathsf{tw}}
\newcommand{\w}{\mathsf{w}}
\newcommand{\leader}{\mathsf{leader}}
\newcommand{\diam}{\mathsf{diam}}
\newcommand{\parent}{\mathsf{parent}}

\newcommand{\depth}{\mathsf{depth}}

\newcommand{\cP}{\mathcal{P}}
\newcommand{\cC}{\mathcal{C}}
\newcommand{\selected}{\mathsf{Selected}}
\newcommand{\Gbase}{G^{\mathsf{base}}}
\newcommand{\Geq}[1]{G^{=#1}}
\newcommand{\Gleq}[1]{G^{\leq#1}}
\newcommand{\OPT}{\mathsf{OPT}}
\newcommand{\ARGOPT}{\mathsf{ARGOPT}}

\newcommand{\Mark}{\mathsf{Mark}}

\newcommand{\CONGEST}{{\rm \textsf{CONGEST}}}

\usepackage[appendix=append]{apxproof}

\title{Distributed Model Checking on Graphs of Bounded Treedepth}

\author{Fedor V. Fomin}{University of Bergen, Bergen, Norway}{Fedor.Fomin@uib.no}{}{}{}
\author{Pierre Fraigniaud}{Université Paris Cité and CNRS, IRIF, Paris, France}{pierre.fraigniaud@irif.fr}{}{}{}
\author{Pedro Montealegre}{Universidad Adolfo Ibañez, Santiago, Chile}{p.montealegre@uai.cl}{}{}{}
\author{Ivan Rapaport}{Universidad de Chile, DIM-CMM, Santiago, Chile}{rapaport@dim.uchile.cl}{}{}{}
\author{Ioan Todinca}{Université d'Orléans and INSA Centre-Val de Loire, LIFO, Orléans, France}{Ioan.Todinca@univ-orleans.fr}{https://orcid.org/0000-0002-3466-859X}{}{}
\authorrunning{Fedor V. Fomin, Pierre Fraigniaud, Pedro Montealegre, Ivan Rapaport, Ioan Todinca} 

\Copyright{Fedor V. Fomin, Pierre Fraigniaud, Pedro Montealegre, Ivan Rapaport, Ioan Todinca} 



\hideLIPIcs
\nolinenumbers
\begin{document}
\maketitle

\begin{abstract}
We establish that every monadic second-order logic (MSO) formula on graphs with bounded treedepth is decidable in a constant number of rounds within the \textsf{CONGEST} model. To our knowledge, this marks the first meta-theorem regarding distributed model-checking. Various optimization problems on graphs are expressible in MSO. Examples include determining whether a graph $G$ has a clique of size $k$, whether it admits a coloring with $k$~colors, whether it contains a graph $H$ as a subgraph or minor, or whether terminal vertices in $G$ could be connected via vertex-disjoint paths. Our meta-theorem significantly enhances the work of Bousquet et al. [PODC 2022], which was focused on \emph{distributed certification} of MSO on graphs with bounded treedepth. 
Moreover, our results can be extended to solving optimization and counting problems expressible in MSO, in graphs of bounded treedepth.
\end{abstract}

\section{Introduction}

Distributed \emph{decision}~\cite{SarmaHKKNPPW12,FraigniaudGKPP14,FraigniaudGKS13} and distributed \emph{certification}~\cite{FraigniaudKP13,GoosS16,KormanKP10} are two complementary fields of distributed computing, closely associated  with distributed fault-tolerant computing~\cite{BlinFB23,Feuilloley21,KormanKM15}. Both fields are addressing the problem of checking whether a distributed system is in a legal state with respect to a given specification, or not. We examine this problem in the classical context of distributed computing in networks, under the standard \CONGEST\/ model~\cite{Peleg2000}. Recall that this model assumes networks modeled as simple connected $n$-node graphs, in which every node is provided with an identifier on $O(\log n)$ bits that is unique in the network. Computation proceeds synchronously as a sequence of \emph{rounds}. At each round, every node sends a message to each of its neighbors in the graph, receives the messages sent by its neighbors, and performs some individual computation. A crucial point is that messages are restricted to be of size $O(\log n)$ bits. 
 While this suffices to transmit an identifier, or a constant number of identifiers, transmitting large messages may require multiple rounds, typically $\Theta(k/\log n)$ rounds for $k$-bit messages.

\subsubsection*{Distributed Decision.} 

Given a boolean predicate~$\Pi$ on graphs, e.g., whether the graph is $H$-free for some fixed graph~$H$, a \emph{decision} algorithm for~$\Pi$ takes as input a graph $G=(V,E)$, and outputs whether $G$ satisfies $\Pi$ or not. Specifically, every node~$v$ receives as input its identifier $\id(v)$, and, after a certain number of rounds of communication with its neighbors, it outputs \emph{accept} or \emph{reject}, under the constraint that $G$ satisfies $\Pi$ if and only if the output of each of the nodes $v\in V$ is accept. In other words, 
\[
G \models \Pi \iff \forall v\in V(G), \; \mbox{output}(v)=\mbox{accept}.
\]
Some predicates are easy to decide  \emph{locally}, i.e., in a constant number of rounds. A canonical example is checking whether the (connected) graph $G$ is regular,  for which one round suffices. However, other predicates cannot be checked locally, with canonical example checking whether there is a unique node of degree~3 in the network. Indeed, checking this property requires $\Omega(n)$ rounds in networks of diameter $\Theta(n)$, as two nodes of degree~3 may be at arbitrarily large distances in the graph. Another example of a difficult problem is checking whether the graph is $C_4$-free, i.e., does not contain a 4-cycle as a subgraph, which requires $\tilde{\Omega}(\sqrt{n})$ rounds~\cite{DruckerKO13}.  One way to circumvent the difficulty of local checkability, i.e., to address graph predicates requiring a large number of rounds for being decided,  is to consider distributed \emph{certification}. 

\subsubsection*{Distributed Certification.} 

A \emph{certification scheme} for a boolean predicate~$\Pi$ is a pair \emph{prover-verifier} (see~\cite{Feuilloley21} for more details). The prover is a centralized, computationally unbounded, non-trustable oracle. Given a graph~$G=(V,E)$, the prover assigns a \emph{certificate} $c(v)\in\{0,1\}^\star$ to each node $v\in V$. These certificates are forged by the prover using the complete knowledge of the graph~$G$. The verifier is a distributed 1-round algorithm. Each node~$v$ takes as sole input its identifier $\id(v)$ and its certificate $c(v)$. In particular, for distributed decisions, $v$ is unaware of the graph~$G$. Every node~$v$ just exchanges once its identifier and certificate with its neighbors, and then it must output \emph{accept} or \emph{reject}. 

The certification scheme is correct if the following two conditions hold. The \emph{completeness} condition states that if $G$ satisfies~$\Pi$, then the oracle can provide the nodes with certificates that they all accept. The \emph{soundness} condition says that if $G$ does not satisfy~$\Pi$, then no matter the certificates assigned by the oracle to the nodes, at least one of them rejects.  That is, the role of the verifier is to check that the collection of certificates assigned to the nodes by the prover is indeed a global proof that the graph satisfies the predicate. In other words, 
\[
G \models \Pi \iff \exists c:V(G)\to \{0,1\}^\star \;:\; \forall v\in V(G),  \; \mbox{output}(v)=\mbox{accept}.
\]
The main measure of complexity of a certification scheme is the maximum \emph{size} of the certificates assigned by the prover to the nodes on legal instances, i.e., for graphs $G$ satisfying the given predicate. 
Ideally, to be implemented in a single round under the \CONGEST\/ model, the certificates should be of size $O(\log n)$ bits. Interestingly, many graph properties can be certified with such short certificates, including acyclicity~\cite{KormanKP10}, planarity~\cite{FeuilloleyFMRRT21},  bounded genus~\cite{EsperetL22,FeuilloleyF0RRT23}, etc. On the other hand, basic graph properties require large certificates, including diameter 2~vs.~3 (requiring $\tilde{\Omega}(n)$-bit certificates~\cite{Censor-HillelPP20}), non-3-colorability (requiring $\tilde{\Omega}(n^2)$-bit certificates~\cite{GoosS16}), $C_4$-freeness (requiring $\tilde{\Omega}(\sqrt{n})$-bit certificates~\cite{DruckerKO13}), etc. The following question was thus raised, under different formulations (see, e.g., \cite{FeuilloleyBP22}): \emph{What are the graph properties that admit certification schemes with $O(\log n)$-bit certificates, or, to the least, certificates of polylogarithmic size?} Answering this question requires formalizing the notion of ``graph predicate''.

\subsubsection*{Monadic Second-Order Logic.} 

Recall that, in the \emph{first-order} logic (FO) of graphs, a graph property is expressed as a quantified logical sentence whose variables represent vertices, with predicates for equality ($=$) and adjacency ($\adj$). An FO formula is therefore constructed according to the following set of rules, where $x$ and $y$ are vertices, and $\varphi$ and $\psi$ are FO formulas:  
\[
x=y \mid \adj(x,y) \mid \varphi\lor\psi \mid  \varphi\land\psi \mid \neg \varphi \mid \exists x \varphi \mid \forall x \varphi.
\]
As an example, triangle-freeness can be formulated as 
\[
\varphi = \neg  \exists x_1 \exists  x_2 \exists  x_3 \big (  \adj(x_1,x_2) \land \adj(x_2,x_3) \land \adj(x_3,x_1) \big ).
\]
The formula above assumes simple graphs (i.e., no loops nor multiple edges). If the graphs may have loops, then one should add the predicate $\neg(x_i=x_j)$ to the formula for every $i\neq j$. 

The \emph{monadic second-order} logic (MSO) extends FO by allowing quantification on \emph{sets} of vertices and edges, with the incidence predicate $\inc(v,e)$ indicating whether vertex $v$ is incident to edge $e$, and the membership ($\in$) predicate.
For instance, acyclicity can be formulated~as 
\[
\varphi = \neg  \exists X\neq \varnothing  \; \forall x\in X \; \exists y_1 \in X \; \exists y_2 \in X \big (  \neg (y_1=y_2) \land \adj(x,y_1) \land \adj(x,y_2) \big ).
\]
Note that $X\neq \varnothing$ can merely be written as $\exists x\in X$. Note also that acyclicity cannot be expressed in FO as the length of the potential cycle is unbounded, from which it follows that one cannot quantify on vertices only for expressing acyclicity, because one does not know \emph{how many} vertices should be considered. On the other hand, since FO can express properties such as $C_4$-freeness, which, as mentioned before, requires certificates on $\tilde{\Omega}(\sqrt{n})$ bits, there is no hope of establishing a \emph{meta-theorem} about FO regarding compact certification in all graphs. Nevertheless, a breakthrough in the theory of distributed certification was recently obtained by  Bousquet, Feuilloley, and Pierron~\cite{FeuilloleyBP22}, who showed that every MSO predicate admits a distributed certification scheme with $O(\log n)$-bit certificates in the family of graphs with bounded \emph{treedepth}.

\subsubsection*{Algorithmic Meta-Theorems.}

A vibrant line of research in sequential computing is the development of {algorithmic meta-theorems}. According to Grohe and Kreutzer \cite{GroheK09}, algorithmic meta-theorems assert that certain families of algorithmic problems, typically defined by some logical and combinatorial conditions, can be solved efficiently under some suitable definition of this term. Such theorems play an essential role in the theory of algorithms as they reveal a profound interplay between algorithms, logic, and combinatorics. One of the most celebrated examples of a meta-theorem is Courcelle's theorem, which asserts that graph properties definable in MSO are decidable in linear time on graphs of bounded treewidth \cite{Courcelle90}. For an introduction to this fascinating research area, we refer to the surveys by Kreutzer \cite{Kreutzer11algo} and Grohe \cite{Grohe07logi}. 

Bousquet, Feuilloley, and Pierron in \cite{FeuilloleyBP22} introduced the exploration of algorithmic meta-theorems in distributed computing. Their primary result in this direction is that any MSO formula can be locally \emph{certified} on graphs with bounded treedepth using a logarithmic number of bits per node, which represents the golden standard in certification. This theorem has numerous consequences for certification --- for more details, we refer to \cite{FeuilloleyBP22}. Notably, the FO property $C_4$-freeness, and the MSO property non-3-colorability, which both necessitate certificates of polynomial size in general, can be certified with just $O(\log n)$-bit certificates in graphs of bounded treedepth. Bousquet et al.'s result has been extended to more comprehensive classes of graphs, including graphs excluding a small minor \cite{BousquetFP21}, as well as graphs of bounded \emph{treewidth}, and graphs of bounded \emph{cliquewidth}. However, this extension comes at the cost of slightly larger certificates, of $O(\log^2n)$ bits, as seen in \cite{FraigniaudMRT24} and \cite{FraigniaudM0RT23}, respectively.

With significant advances in developing meta-theorems for distributed \emph{certification}, there's a notable absence of similar results for distributed \emph{decision}. It prompts a natural question: could such results be obtained for the round-complexity of \CONGEST? More concretely, the fundamental inquiry that remains unaddressed by Bousquet et al.'s paper, and by the subsequent works regarding distributed certification of MSO predicates is:

\begin{tcolorbox}[colback=green!5!white,colframe=gray!75!black]
\medbreak
\noindent\textbf{Question.} \emph{What is the round-complexity in \CONGEST\ of deciding MSO formulas in graphs of bounded treedepth?}
\medbreak
\end{tcolorbox}

A first step in answering this question was proposed in~\cite{NesetrilM16} where it is stated that, in any graph class of treedepth at most $d$, for every fixed connected graph~$H$, $H$-freeness can be decided in $O(1)$ rounds in \CONGEST. In this paper, we offer a comprehensive answer to the question. To elucidate our results, we first need to define the treedepth of a graph. 

\subsubsection*{Treedepth.} 

For any non-negative integer~$d$, a (connected) graph~$G$ has treedepth at most~$d$ if there exists a rooted tree~$T$ spanning the vertices of~$G$, with depth at most~$d$, such that, for every edge $\{u,v\}$ in~$G$, $u$ is an ancestor of $v$ in $T$, or $v$ is an ancestor of $u$ in $T$, 
see Fig.~\ref{fig:treedepthb}. The treedepth of a graph~$G$, denoted by $\td(G)$, is the smallest~$d$ for which such a tree exists. 

\begin{figure}
\includegraphics[scale=.3]{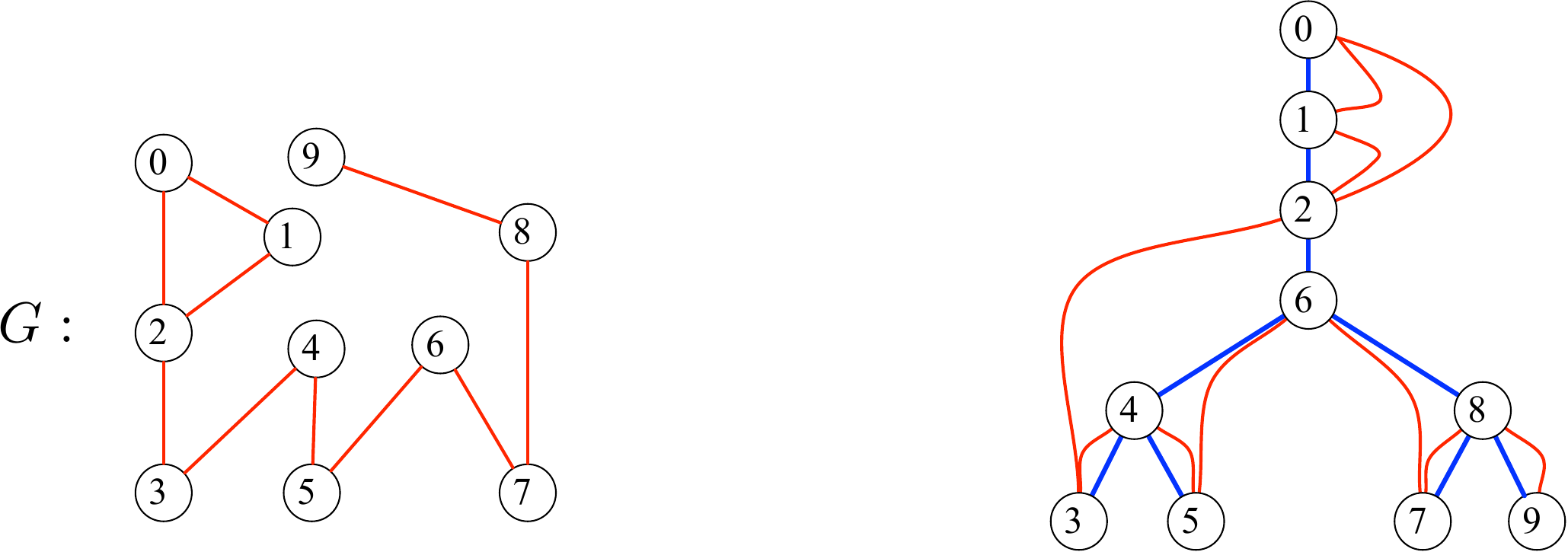}
\caption{Embedding of a graph $G$ into a tree $T$ of depth 6.}\label{fig:treedepthb}
\end{figure}

The class of graphs with bounded treedepth, i.e., of treedepth~$d$ for some fixed $d\geq 0$, has strong connections with minor-closed families of graphs. Specifically, for any family $\mathcal {F}$ of graphs closed under taking graph minors,  the graphs in $\mathcal {F}$ have bounded treedepth if and only if $\mathcal {F}$ does not include all the paths~\cite{NesetrilOdM12}. Similarly, the graphs with bounded 
treedepth have a finite set of forbidden induced subgraphs, and any property of graphs monotonic with respect to induced subgraphs can be tested in polynomial time on graphs of bounded treedepth~\cite{NesetrilOdM12}. Computing the treedepth of a graph is NP-hard, but since treedepth is monotonic under graph minors, it is fixed-parameter tractable (FPT)~\cite{GajarskyH15}. 
Last but not least, MSO and FO have the same expressive power in graph classes of bounded treedepth~\cite{ElberfeldGT16}.  

\subsection{Our Results}

We prove that, for every MSO formula~$\varphi$, there is an algorithm~$\mathcal{A}$ that, for every $n$-node graph~$G$, decides whether $G\models \varphi$ in $O(2^{2\td(G)})$ rounds in the \CONGEST\/ model.  That is, the round-complexity of $\mathcal{A}$ depends only on the treedepth of the input graph, and on the MSO formula, i.e., it does not depend on the \emph{size}~$n$ of the graph. Thus $\mathcal{A}$ performs a constant number of rounds in any class of graphs with bounded treedepth. In particular, deciding non-3-colorability can be done in $O(1)$ rounds in graphs of bounded treedepth, in contrast to general graphs, for which deciding non-3-colorability requires a polynomial number of rounds by~\cite{GoosS16}.

Our meta-theorem is essentially the best that one may hope regarding distributed model checking MSO formulas in a constant number of rounds in \CONGEST. Indeed, the FO predicate ``there is at least one vertex of degree~$>2$'' requires $\Omega(n)$ rounds to be checked in this class. 
Hence our theorem cannot be extended to graphs of bounded treewidth or bounded cliquewith, actually not even to bounded pathwith, and not even to the class $\mathcal{P}\cup\mathcal{B}$ where $\mathcal{P}$ is the set of all paths, and $\mathcal{B}$ is the set of all graphs composed of a path to which is attached a claw at one of its endpoints.
 
We also consider distributed model checking of \emph{labeled} graphs. For instance, one can check whether a given set of vertices is a feedback vertex set, i.e., whether the graph obtained by removing this set of vertices is acyclic. For such a predicate, it is sufficient to add a unary predicate to the logical structure used to mark the nodes, say $\mathsf{mark}(x)=true$ means that vertex $x$ is in the set. Using this unary predicate, $\varphi$ can express the fact that there are no cycles in $G$ passing only trough nodes $x$ for which $\mathsf{mark}(x)=false$. As another example, the fact that a graph is properly 2-colored can be expressed using two unary boolean predicates $\mathsf{red}$ and $\mathsf{blue}$, as 
\[
\varphi= 
 \Big (\forall x \; \big (\mathsf{red}(x)\lor \mathsf{blue}(x) \big )\Big) \land 
 \bigg ( \forall x,y \; \neg \Big( \adj(x,y) \land 
\big ((\mathsf{red}(x)\land \mathsf{red}(y)) \lor (\mathsf{blue}(x)\land \mathsf{blue}(y) \big)\Big )
\bigg ).
\]
Since we also deal with $\mbox{MSO}$, we can also label edges. For instance, one can check whether a given set of edges forms a spanning tree.  Indeed, it is sufficient to introduce a unary predicate used to mark the edges: $\mathsf{mark}(e)=true$ means that edge~$e$ is in the set. As for feedback vertex set, using this unary predicate,  $\varphi$ can express the fact that the set of marked edges is a spanning tree (i.e., every node is incident to at least one marked edge, and any two vertices are connected by a path of marked edges). We show that deciding  MSO formulas on \emph{labeled} graphs of bounded treedepth can be done in $O(1)$ rounds in the \CONGEST\/ model.

More generally, we also consider the \emph{optimization} variants of decision problems expressible in MSO on graphs of bounded treedepth. For instance, an independent set can be expressed as an MSO formula with a free variable~$S$, such as
$
\varphi(S) = \forall x\in S \; \forall y\in S \; \neg\adj(x,y).
$
Then, $max\varphi$, i.e., maximum independent set, consists in, given any graph~$G=(V,E)$, finding the largest set $S\subseteq V$ such that $G\models \varphi(S)$. We show that, for every MSO formula $\varphi(S)$ with free variable $S\subseteq V$ or $S\subseteq E$, there is an algorithm for graphs of bounded treedepth solving $max\varphi$ (and $min\varphi$) in a constant number of rounds in the \CONGEST\/ model. This constant is of the form $O(g(\td(G),\varphi))$ for some function $g$.
Due to the expressive power of MSO, our results yield algorithms with a constant number of rounds in the \CONGEST\/ model on graphs of bounded treedepth for numerous popular optimization problems including minimum vertex cover, minimum feedback vertex set, minimum dominating set, maximum independent set, maximum induced forest, maximum clique, maximum matching, minimum spanning tree, Hamiltonian cycle, cubic subgraph, planar subgraph, Eulerian subgraph, Steiner tree, disjoint paths, min-cut, minor and topological minor containment, rural postman, $k$-colorability, edge $k$-colorability, partition into $k$ cliques, and covering by $k$ cliques. We also extend our results to \emph{counting} problems, such as counting triangles or perfect matchings.

Finally, we briefly discuss some applications of our results to much larger classes of graphs, namely graphs of bounded \emph{expansion} (see~\cite{NesetrilOdM12} for an extended introduction). Graphs of bounded expansion include planar graphs, and more generally, all classes of graphs defined from excluding minor. It was shown~\cite{NesetrilM16} that, for every class $\mathcal{G}$ of graphs with bounded expansion, and every positive integer~$p$, there is an algorithms performing in $O(\log n)$ rounds under the \CONGEST\ model that partitions the vertex set~$V$ of any graph $G=(V,E)\in \mathcal{G}$ into $f(p)$ parts $V_1,\dots,V_{f(p)}$ such that every collection $V_{i_1},\dots,V_{i_q}$ of at most $p$ parts, $1\leq q \leq p$, $\{i_1,\dots,i_q\}\subseteq \{1,\dots,f(p)\}$, induces a (not necessarily connected) subgraph of $G$ with treedepth at most $p$. The function $f$ solely depends on the considered class~$\mathcal{G}$ of bounded expansion. The vertex partitioning $V_1,\dots,V_{f(p)}$ is called a low treedepth decomposition with parameter~$p$.
Plugging in our techniques into this framework, we show that, for every connected graph~$H$, $H$-freeness can be decided in $O(\log n)$ rounds under the \CONGEST\ model in any class of graphs with bounded expansion. This result was claimed in~\cite{NesetrilM16} with no proofs. We provide that claim with a complete formal proof. 

\subsection{Other Related Work}

The quest for efficient (sublinear) algorithms for solving classical graph problems in the CONGEST model  dates back to the seminal paper by Garay, Kutten and Peleg \cite{garay1993sub}, where 
an algorithm for constructing an MST was designed. Since then, a long series of problems have been addressed, such as 
connectivity decomposition \cite{censor2014distributed},
tree embeddings \cite{ghaffari2014near}
$k$-dominating set \cite{kutten1995fast},
stiener trees \cite{lenzen2014improved},
 min-cut \cite{ghaffari2013distributed,nanongkai2014almost}, max-flow \cite{ghaffari2015near}, shortest path \cite{henzinger2016deterministic,nanongkai2014distributed}, among others.
Additionally, algorithms tailored to specific classes of networks have also been developed: DFS for planar graphs \cite{ghaffari:2017}, MST for bounded genus graphs \cite{haeupler2016low}, MIS for networks excluding a fixed minor \cite{chang2023efficient}, etc.

Distributed certification is a very vivid topic, and many results have appeared since the survey~\cite{FeuilloleyF16}. A handful of recent papers considered \emph{approximate} variants of the problem, a la property testing \cite{Elek22,EmekGK22,FeuilloleyF22}. In particular, it was shown that every monotone (i.e., closed under taking subgraphs) and summable (i.e., stable by disjoint union) property~$\Pi$ has a compact approximate certification scheme in any proper minor-closed graph class~\cite{EsperetN22}. Other recent contributions dealt with augmenting the power of the verifier in certification schemes, which includes tradeoffs between the size of the certificates and the number of rounds of the verification protocol~\cite{FeuilloleyFHPP21}, randomized verifiers~\cite{FraigniaudPP19}, quantum verifiers~\cite{FraigniaudGNP21}, and several hierarchies of certification mechanisms, including games between a prover and a disprover~\cite{BalliuDFO18,FeuilloleyFH21}, interactive protocols~\cite{CrescenziFP19,KolOS18,NaorPY20}, and even zero-knowledge distributed certification~\cite{BickKO22}, and distributed quantum interactive protocols~\cite{GallMN23}.


\section{Treedepth and treewidth}

Throughout the paper, trees (or forests) are considered as rooted. The \emph{depth} of a tree is the number of vertices of a longest path from the root to a leaf. The depth of a forest is the maximum depth  among its trees.
Let us recall the definition of the treedepth. The interested reader can refer to the book of Ne\v{s}et\v{r}il and Ossona de Mendez~\cite{NesetrilOdM12} for further insights.

\begin{definition}[treedepth]
The \emph{treedepth} of a graph $G=(V,E)$ is the minimum depth of a forest $T=(V,F)$ on the same vertex set as~$G$, such that, for any edge $\{u,v\}$ of $G$, one of the endpoints is an ancestor of the other in the forest $T$.
Such a forest $T$ is also called \emph{elimination forest} of $G$.
\end{definition}

Observe that if $G$ is connected then the forest $T$ in the definition above is actually a tree.  
The following statement is an alternative, equivalent definition for treedepth. This recursive definition implicitly provides a recursive construction of an elimination tree.

\begin{lemma}[\cite{NesetrilOdM12}]\label{le:tdAlter}
The treedepth of a graph $G$ is:
\begin{equation*}
\td(G)=
\begin{cases}
1 &\mbox{if $G$ has a unique vertex},\\
1+\min_{v\in V(G)}\td(G-v) &\mbox{if }  G\text{ is connected},\\
\max\{\td(C)\mid C\text{ is a connected component of }G\} &\mbox{otherwise}.
\end{cases}
\end{equation*}
\end{lemma}

\noindent On the other hand, tree decompositions and treewidth of graphs were introduced by Robertson and Seymour~\cite{RoSe84}.

\begin{definition}[treewidth]\label{de:treedec} 
A \emph{tree decomposition} of a graph $G = (V,E)$ is a pair $(T, B)$  where ${T=(I,F)}$ is a tree, 
and $B=\{B_i, i \in I\}$ is a collection of subsets of vertices of $G$, called \emph{bags}, such that the following conditions hold:
\begin{itemize}
\item For every vertex of $G$, there exists some bag containing this vertex;
\item For every edge $e$ of $G$ there is some bag containing both endpoints of $e$;
\item For every $v \in V$, the set $\{i \in I : v \in B_i\}$ of bags containing $v$ forms a connected subgraph of $T$. 
\end{itemize}
The \emph{width} of a tree decomposition is the maximum size of a bag, minus one. The \emph{treewidth} of a graph~$G$, denoted by~$\tw(G)$, is the smallest width of a tree decomposition of~$G$.
\end{definition}


It is known~\cite{NesetrilOdM12} that the treedepth of a graph is at least its treewidth. Given an elimination tree $T$ of a graph $G$, we can define a canonical tree decomposition associated to this same tree, such that the width of the decomposition corresponds to the depth of $T$, minus one. The following lemma is a straightforward consequence of the definitions of elimination trees and of tree decompositions.

\begin{lemma}[canonical tree decomposition] Let $T=(V,F)$ be an elimination tree of depth $d$ of a graph $G=(V,E)$. Let us associate to each node $u$ of $T$ a bag $B(u)$ containing $u$ and all the ancestors of $u$ in $T$. 
Then $T=(V,F)$, and the corresponding set of bags $(B_u)_{u \in V}$, form a tree-decomposition of $G$, of width $d-1$.
\end{lemma}

For instance, the treedepth of an $n$-vertex path $P_n$ is exactly $\lceil \log(n+1)  \rceil$ (see, e.g.,~\cite{NesetrilOdM12}).
The treedepth of a graph does not increase when we delete some of its  edges or vertices. Thus graphs of treedepth~$d$ have no paths on $2^d$ vertices. This observation yields the following lemma. 


\begin{lemma}\label{le:greedyDFS}
    Let $T=(V,F)$ be an elimination tree of a graph $G=(V,E)$ with $F \subseteq E$. Then the depth of $T$ is at most $2^{\td(G)}$. 
\end{lemma}
\begin{proof}
%
    Let $d=\td(G)$. Assume, for the purpose of contradiction, that $T$ has depth larger than~$2^d$. It follows that the longest path~$P$ in~$T$ from its root to a leaf contains at least $2^d$ vertices. The path $P$ is also a path in~$G$, so the treedepth of $P$ is at most the treedepth of~$G$, i.e., at most $d$. This is a contradiction with the fact that, for $n$-node paths, $\td(P_n) = \lceil \log(n+1)\rceil$.
\end{proof}

\section{Tree decompositions and $w$-terminal recursive graphs}

Courcelle's theorem~\cite{Courcelle90} states that any property expressible in MSO can be decided in linear (sequential) time on graphs of bounded treewidth. We use an alternative proof of Courcelle's theorem, by Borie, Parker, and Tovey~\cite{BoPaTo92}. Indeed, this proof provides us with an explicit dynamic programming strategy, which will be used in our distributed protocol.


Graphs  of bounded treewidth can also be defined recursively, based on a graph grammar. Let $w$ be a positive integer. A \emph{$w$-terminal graph} is a triple $(V,W,E)$ where $G=(V,E)$ is a graph, and $W \subseteq V$ is a \emph{totally ordered} set  of at most $w$ distinguished vertices. Vertices of $W$ are called the \emph{terminals} of the graph, and we denote by $\tau(G)$ the number of its terminals. As the terminal set $W$ is totally ordered, we can refer to the $r$th terminal, for $1\leq r \leq w$.  Moreover, since vertices are given distinct identifiers in \CONGEST, one can view $W$ as ordered by these identifiers.

The class of $w$-terminal \emph{recursive} graphs is defined, starting from $w$-terminal \emph{base} graphs, by a sequence of \emph{compositions}, or \emph{gluings}. A $w$-terminal base graph is a $w$-terminal graph of the form $(V,W,E)$ with $W=V$. A composition $f$ acts on two\footnote{The definition of~\cite{BoPaTo92} considers composition operations of arbitrary arity, i.e., they consider gluing on three or more graphs simultaneously, and they also consider a special gluing, on a single graph which allows to ``forget'' some terminals.  Technically, all these operations can be replaced by operations of arity~2, and we only use arity~2 for the sake of simplifying the presentation.} $w$-terminal graphs, and  produces a new $w$-terminal graph, as follows (see Figure~\ref{fig:SimpleTwGram} for an example\footnote{The figure is borrowed from~\cite{FraigniaudM0RT23} with the agreement of the authors.}, for $w=2$). 

The graph $G = f(G_1, G_2)$ is obtained by, first, making disjoint copies of the two graphs $G_1$ and $G_2$, and, second, ``gluing'' together some terminals of $G_1$ and $G_2$. In the gluing operation, each terminal of \(G_1\) is identified with at most one terminal of~\(G_2\). Formally, the gluing performed by  $f$ is represented by a matrix $m(f)$ having $\tau(G) \leq w$ rows, and two columns, with integer entries in $\{0,\dots,\tau(G)\}$. At row $r$ of the matrix, $m_{r,c}(f)$ indicates which terminal of each graph $G_c$, $c \in \{1,2\}$ is identified to the $r$th terminal of graph $G$. If $m_{r,c}(f) = 0$, then no terminal of $G_c$ is identified to terminal $r$ of $G$ (in particular, if  $m_{r,1}(f) = m_{r,2}(f) = 0$, then the $r$th terminal of $G$ is a new vertex; Nevertheless, this situation will not occur in our construction). Every terminal of $G_c$ is identified to at most one terminal of~$G$, i.e., each non-zero value in $\{1,\dots, \tau(G_c)\}$  appears at most once in the column $c$ of~$m(f)$. 


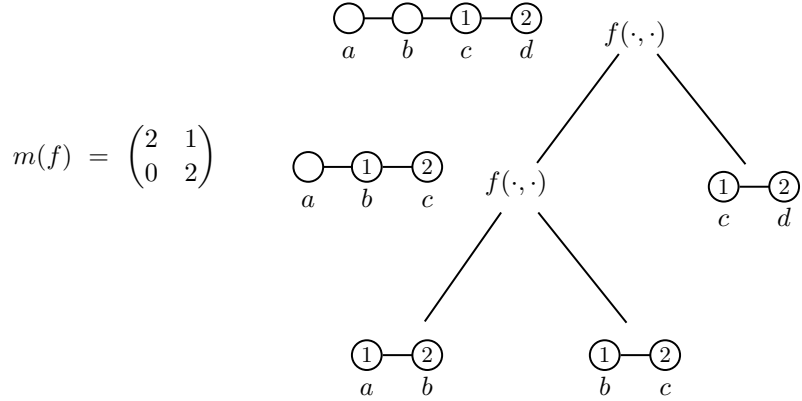
\begin{figure}[htbp]
\begin{center}
\tikzset{every picture/.style={line width=0.75pt}} 

\begin{tikzpicture}[x=0.75pt,y=0.75pt,yscale=-1,xscale=1]

\draw (460,60) node [anchor=north west][inner sep=0.75pt]    {$f( \cdot ,\cdot )$};

\draw (400,135) node [anchor=north west][inner sep=0.75pt]    {$f( \cdot ,\cdot )$};
\draw    (469.13,78) -- (427.88,133) ;
\draw    (533,133.63) -- (488.5,78) ;
\draw    (371.25,213) -- (409.75,158) ;
\draw    (473,213.63) -- (428.5,158) ;

\draw    (30+310.5,60) -- (30+325,60) ;
\draw    (30+340,60) -- (30+355,60) ;
\draw    (30+370,60) -- (30+385,60) ;

\draw (30+303,60) circle [x radius=  7.5, y radius= 7.5]   ;
\draw (30+303,60) node  [font=\footnotesize]  {$$};
\draw (30+303,85) node [anchor=south]{$a$};
\draw (30+332.5,60) circle [x radius=  7.5, y radius= 7.5]    ;
\draw (30+332.5,60) node  [font=\footnotesize]  {$$};
\draw (30+332.5,85) node [anchor=south]{$b$};
\draw (30+362, 60) circle [x radius=  7.5, y radius= 7.5]   ;
\draw (30+362, 60) node  [font=\footnotesize]  {$1$};
\draw (30+362,85) node [anchor=south] {$c$};
\draw (30+392.5,60) circle [x radius=  7.5, y radius= 7.5]    ;
\draw (30+392.5,60) node  [font=\footnotesize]  {$2$};
\draw (30+392.5,85) node [anchor=south]{$d$};

\draw    (30+340-50,135) -- (30+355-50,135) ;
\draw    (30+370-50,135) -- (30+385-50,135) ;

\draw (30+332.5-50,135) circle [x radius=  7.5, y radius= 7.5]    ;
\draw (30+332.5-50, 135) node  [font=\footnotesize]  {$$};
\draw (30+332.5-50, 160) node [anchor=south]{$a$};
\draw (30+362-50, 135) circle [x radius=  7.5, y radius= 7.5]   ;
\draw (30+362-50, 135) node  [font=\footnotesize]  {$1$};
\draw (30+362-50, 160) node [anchor=south] {$b$};
\draw (30+392.5-50,135) circle [x radius=  7.5, y radius= 7.5]    ;
\draw (30+392.5-50, 135) node  [font=\footnotesize]  {$2$};
\draw (30+392.5-50, 160) node [anchor=south]{$c$};

\draw    (30+370-50,230) -- (30+385-50,230) ;

\draw (30+362-50, 230) circle [x radius=  7.5, y radius= 7.5]   ;
\draw (30+362-50, 230) node  [font=\footnotesize]  {$1$};
\draw (30+362-50, 255) node [anchor=south] {$a$};
\draw (30+392.5-50, 230) circle [x radius=  7.5, y radius= 7.5]    ;
\draw (30+392.5-50, 230) node  [font=\footnotesize]  {$2$};
\draw (30+392.5-50, 255) node [anchor=south]{$b$};

\draw    (150+370-50,230) -- (150+385-50,230) ;

\draw (150+362-50, 230) circle [x radius=  7.5, y radius= 7.5]   ;
\draw (150+362-50, 230) node  [font=\footnotesize]  {$1$};
\draw (150+362-50, 255) node [anchor=south] {$b$};
\draw (150+392.5-50, 230) circle [x radius=  7.5, y radius= 7.5]    ;
\draw (150+392.5-50, 230) node  [font=\footnotesize]  {$2$};
\draw (150+392.5-50, 255) node [anchor=south]{$c$};

\draw    (210+370-50,145) -- (210+385-50,145) ;

\draw (210+362-50, 145) circle [x radius=  7.5, y radius= 7.5]   ;
\draw (210+362-50, 145) node  [font=\footnotesize]  {$1$};
\draw (210+362-50, 170) node [anchor=south] {$c$};
\draw (210+392.5-50, 145) circle [x radius=  7.5, y radius= 7.5]    ;
\draw (210+392.5-50, 145) node  [font=\footnotesize]  {$2$};
\draw (210+392.5-50, 170) node [anchor=south]{$d$};

\draw (215,130) node {$m( f) \ =\ \begin{pmatrix}
2 & 1\\
0 & 2
\end{pmatrix}$};

\end{tikzpicture}
\caption{Paths as 2-terminal recursive graphs.}
\label{fig:SimpleTwGram}
\end{center}
\end{figure}

A simple but crucial observation is that the number of possible different matrices, and hence of different composition operations~$f$, is bounded by a function of~$w$. Indeed the size of each matrix is at most~$2w$, and each entry of the matrix is an integer between $0$ and~$w$.

The class of $w$-terminal recursive graphs is exactly the class of graphs of treewidth at most $w-1$ (see, e.g., Theorem~40 in~\cite{Bodlaender98arb}). 

Let us briefly describe how a tree-decomposition of width $w-1$ of a graph $G$ can be transformed into a description of $G$ as a $w$-terminal recursive graph.  This construction will be crucial for efficiently deciding MSO properties of graph $G$. Let $T=(I,F)$ be a tree-decomposition of $G=(V,E)$ with bags of size at most $w$. The terminals correspond to the root bag. For every node $u$ of $T$, we use the following notations, depicted in Figure~\ref{fig:Gu}:
\begin{itemize}
    \item $T_u$ is the subtree of $T$ rooted at $u$;
    \item $B_u$ is the bag of node $u$, and $\Gbase_u = (G[B_u],B_u)$ is the $w$-terminal recursive base graph induced by bag $B_u$;
    \item $V_u$ is the union of all bags of $T_u$, and $G_u = (G[V_u],B_u)$ is the $w$-terminal graph induced by $V_u$, with $B_u$ as set of terminals. 
\end{itemize}

\begin{figure}
    \centering
    \includegraphics[scale=0.25]{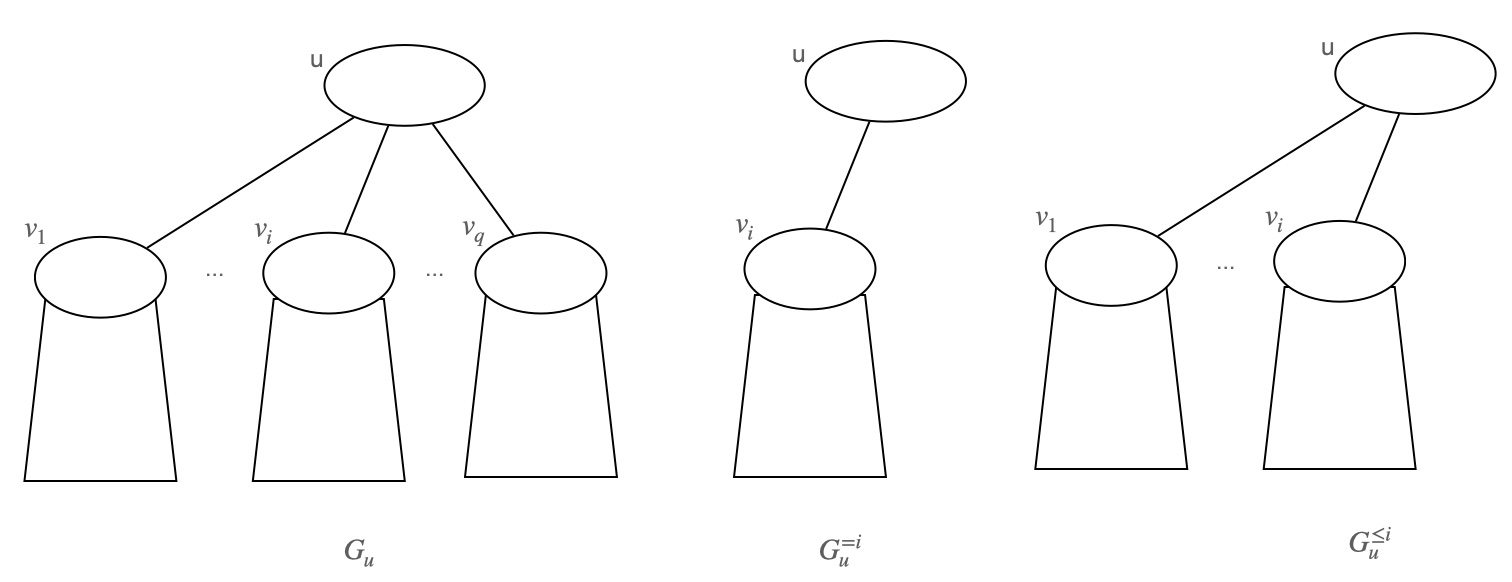}
    \caption{Tree-decompositions: graphs $G_u$, $\Geq{i}_u$ and $\Gleq{i}_u$}
    \label{fig:Gu}
\end{figure}

Let us now show that $G_u$ is indeed a $w$-terminal recursive graph. This is clear when $u$ is a leaf, since, in this case, $G_u = \Gbase_u$ is a base graph.
Assume that $u$ is not a leaf, and let $v_1, \dots, v_q$ be the children of node $u$ in~$T$. The ordering of the children is arbitrary, but fixed. Let us introduce two new families of $w$-terminal recursive graphs as follows. Both are having $B_u$ as set of terminals, and, for every $i\in\{1\dots, q\}$:
\begin{itemize}
    \item $\Geq{i}_u = G[B_u \cup V_{v_i}]$, and 
    \item $\Gleq{i}_u= G[B_u \cup V_{v_1} \cup \dots \cup V_{v_i}]$.
\end{itemize}
Observe that $\Geq{i}_u$ is obtained by gluing $G_{v_i}$ with the base graph~$G_u^{base}$. More precisely, 
\begin{equation}\label{eq:Geq}
 \Geq{i}_u = f_{(B_{v_i},B_u)} (G_{v_i},\Gbase_u), 
\end{equation}
 where the gluing operation $f(B_{v_i},B_u)$ glues the terminals of $B_{v_i} \cap B_u$ of $G_{v_i}$ to the corresponding terminals of~$B_u$, and the new set of terminals is $B_u$.
Also, for all $i\in\{1,\dots,q-1\}$, $\Gleq{i+1}_u$ is obtained by gluing $G_u^{\leq i}$ and $G_u^{=i+1}$ using the gluing function $f_{(B_u,B_u)}$ (which is the identity function on $|B_u|$ terminals), that is:
\begin{equation}\label{eq:Gleq}
 \Gleq{i+1}_u = f_{(B_u,B_u)} (\Gleq{i}_u,\Geq{i+1}_u).
\end{equation}
By construction, we get  $G_u = G_u^{\leq q}$.

\section{MSO logic and Courcelle's theorem}

Recall that, using monadic second-order (MSO) logic formulas on graphs, we can express graph properties such as ``G is not 3-colorable'' or ``G contains no triangles''. In order to solve optimization problems, we also consider MSO formulas with a free variable. That is, we consider formulas of the form $\varphi(S)$ where $S$ is a set of vertices, or a set of edges. The corresponding optimization problem aims at finding a set $S$ with maximum (or minimum, depending on the context) size satisfying $G \models \varphi(S)$. More generally, we may assume that the vertices, or edges of the input graph $G=(V,E)$ have \emph{polynomial} weights, that is the weight assignment $\w : V \cup E \to \mathbb{Z}$ satisfies that, for every $x\in V\cup E$,  $w(x)$ can be encoded on $O(\log n)$ bit. The problem $max\varphi$ then consists to compute the set $S$  with maximum weight satisfying $G \models \varphi(S)$. In this framework we can express problems like maximum (weighted) independent set, minimum (weighted) dominating set, or minimum-weight spanning tree (MST). 

\subsection{Regular Predicates, Homorphism Classes, and Composition}

To start, let us first consider  closed formulae only, i.e., with no free variable,  and formulas with just one free (edge or vertex) set variable. Using closed formulae, we can refer to graph predicates~$\cP(G)$, and, using formulas with free variables, we can refer to graph predicates~$\cP(G,X)$, where $X$ denotes a subset of vertices or a subset of edges of~$G$. For each possible assignment of $X$ with corresponding values, $\cP$ is either true or false.

Any composition operation $f$ over two $w$-terminal recursive graphs $G_1=(V_1,W_1,E_1)$, and  $G_2={(V_2,W_2,E_2)}$ naturally extends to a composition over pairs $(G_1,X_1)$, and 
$(G_2,X_2)$.
If $G=f(G_1,G_2)$, we denote by $\circ_f$ the composition over pairs. More precisely, 
$$\circ_f\big((G_1,X_1),(G_2,X_2)\big) = (G,X),$$ 
the operation being valid only under some specific conditions.
Let us consider the case when $X_1$ and $X_2$ are vertex-set variables. 
For each terminal $t$ of $G$, if terminals from both $G_1$ and $G_2$ were mapped to~$t$, say, terminals $t_1 \in W_1$ and $t_2 \in W_2$, then either $t_1 \in X_1$ and $t_2 \in X_2$, or $t_1 \not\in X_1$ and $t_2 \not\in X_2$. The set $X$ is interpreted as the union of $X_1$ and $X_2$, by identifying pairs of terminal vertices $t_1\in X_1$ and $t_2\in X_2$ that have been mapped on a same terminal of~$G$. For edge-sets, the set $X$ can also be seen as the union of two sets $X_1$ and $X_2$, up to gluing the  vertices specified by $f$. We refer to~\cite{BoPaTo92} for a  description of the gluing operation, and of the interpretation of the values of the variables.


\begin{definition}[regular predicate]\label{de:reg}
A graph predicate $\cP(G,X)$ is \emph{regular} if, for any value~$w$, we can associate to~$w$  
\begin{itemize}
\item a finite set $\cC$ of \emph{homomorphism classes},
\item an \emph{homomorphism function} $h$, assigning to each $w$-terminal recursive graph $G$, and  to any subset $X$ of vertices or edges of $G$, a \emph{class} $h(G,X) \in \cC$, and 
\item an \emph{update function} $\odot_f : \cC \times \cC \to \cC$ for each composition operation $f$, 
\end{itemize}
such that:
\begin{enumerate}
\item If $h(G_1,X_1) = h(G_2,X_2)$ then $\cP(G_1,X_1) = \cP(G_2,X_2)$;
\item For any two $w$-terminal recursive graphs $G_1$ and $G_2$, and any two sets $X_1$ and $X_2$,
 $$h\Big(\circ_f\big((G_1,X_1),(G_2,X_2)\big)\Big) = \odot_f\big(h(G_1,X_1),h(G_2,X_2)\big).$$
\end{enumerate}
\end{definition}

A class $c \in \cC$ is said to be an \emph{accepting class} if there exists $(G,X)$ such that $h(G,X)=c$, and $\cP(G,X)$ is true. By definition, the predicate $\cP$ holds for every $(G',X')$ such that $h(G',X')=c$. A non accepting class $c$ is called a \emph{rejecting class}. 
The same definitions applies to predicates $\cP(G)$, with no free variables.

\noindent
\textbf{Remark.}
Without loss of generality,  we may assume that, in Definition~\ref{de:reg}, the class $c = h(G,X)$ with $G=(V,W,E)$ encodes the intersection of $X$ with~$G[W]$. Indeed, since $W$ is of constant size, if $X$ is a vertex-set, then we can add the set of all the ranks of the elements in $X_j \cap W$, with respect to the totally ordered set $W$, to the class~$c$. And if $X_j$ is an edge-set, then we can store each edge of $X_j$ contained in $G[W]$ as the pair of ranks of its endpoints. 
In particular, we can assume that we are given a function $\selected(c,W)$ which, given a class $c$, and a set of terminals $W$, returns the unique intersection of $X$ with the vertices, or the edges, of $G[W]$.
\medskip


\begin{theorem}[\cite{BoPaTo92}]\label{th:reg}
Any predicate $\cP(G, X)$ expressible by an $MSO$ formula $\varphi(X)$ is regular. Moreover, given the formula $\varphi(X)$ and a parameter $w$, one can compute the set of classes~$\cC$, the update functions $\odot_f$ over all possible composition operations~$f$, as well as the homomorphism classes $h(G, X)$ for all base graphs $G$, and all possible values of variable~$X$. (The same holds for predicates $\cP(G)$ corresponding to closed formulas~$\varphi$.)
\end{theorem}

Let us emphasize that the width parameter $w$, and the formula $\varphi$  in Theorem~\ref{th:reg} are \emph{constants}. Thereofore, the set of homomorphism classes $\cC$ is of constant size, and can be computed, as well as functions $\odot_f$ and homomorphism classes of base graphs, in \emph{constant} time. This constant just depends on~$w$ and on~$\varphi$. 

\subsection{Sequential Model-Checking}

We have all ingredients to describe the key ingredient of the sequential model-checking algorithm on graphs of bounded treewidth that we will use for designing our distributed protocol.  In a nutshell, the algorithm proceeds by dynamic programming, from the leaves of the decomposition-tree to the root. When considering a node~$u$, the program deals with the graph $G_u$ induced by all bags in the subtree rooted at~$u$, which is viewed as a $w$-terminal recursive graph with labels~$B_u$. The program computes the homomorphism class of $h(G_u)$ using merely the homomorphism classes of its children, the bags of its children, and the subgraph of~$G$ induced by the bag of~$u$.

\begin{lemma}[bottom-up decision]\label{le:decision}
Let $\cP(G)$ be a regular predicate on graphs, corresponding to a formula $\varphi$ with no free variables. 
Let $G$ be a graph, and let $T=(I,F)$ be a tree-decomposition of~$G$ with bags $\{B_u\mid u \in I\}$.
Let $u$ be a node of the tree decomposition, with children $v_1,\dots, v_q$ for $q\geq 0$.
The homomorphism class of $h(G_u)$ can be computed using only $\Gbase_u$, the values of $B_{v_i}$ and $h(G_{v_i})$ for all $i\in\{1,\dots,q\}$.
\end{lemma}

\begin{proof}
Observe that $h(\Gbase_u)$ can be computed directly by Theorem~\ref{th:reg}. In particular this settles the case when $u$ is a leaf. If $u$ is an internal node, then,
    for each $i=1,\dots, q$, we can compute $h(\Geq{i}_u)$ using $B_{v_i}, B_u, h(G_{v_i})$ and $h(\Gbase_u)$ as follows. By Equation~\ref{eq:Geq}, we have $\Geq{i}_u = f_{(B_{v_i},B_u)}(G_{v_i},\Gbase_u)$. Since $B_{v_i}$ and $B_u$ are known, one can construct the function $f_{(B_{v_i},B_u)}$, and, thanks to  Theorem~\ref{th:reg}, one can  retrieve the function $\odot_{f_{(B_{v_i},B_u)}}$. By Definition~\ref{de:reg}, $h(\Geq{i}_u) = \odot_{f_{(B_{v_i},B_u)}} (h(G_{v_i}),h(\Gbase_u))$. Since the parameters on the right-hand side of the equality are known, one can compute $h(\Geq{i}_u)$.

    Let us now show how to compute the values $h(\Gleq{i}_u)$. For $i=1$, $\Gleq{1}_u=\Geq{1}_u$, and thus $h(\Gleq{1}_u)=h(\Geq{1}_u)$. For every $i\geq 2$, one can compute $h(\Gleq{i}_u)$ using $B_u$, $h(\Gleq{i-1}_u)$, and $h(\Geq{i}_u)$. Indeed, by Equation~\ref{eq:Gleq}, $\Gleq{i}_u = f_{(B_u,B_u)} (\Gleq{i-1}_u,\Geq{i}_u)$, so again we have $h(\Gleq{i}_u) = \odot_{f_{(B_u,B_u)}}(h(\Gleq{i-1}_u),h(\Geq{i}_u))$, and all parameters on the right-hand side have been computed previously. 

Eventually, since $G_u = \Gleq{q}_u$, we get  $h(G_u) = h(\Gleq{q}_u)$.
\end{proof}

\subsection{Sequential Optimization}

Let us now describe how, given an MSO formula $\varphi(S)$ with a free variable $S$, one can solve the problem $max\varphi$, i.e., finding a set $S$ with maximum weight satisfying $G\models \varphi(S)$. Again, the algorithm proceeds  by dynamic programming, from the leaves of the decomposition tree to the root. Let us consider a node $u$ of the tree decomposition. 
For all homomorphism classes $c\in \cC$, the goal is to maximize the weight of a set $S_u$ such that $h(G_u,S_u)=c$, using the information of the same nature retrieved from the children of $u$ in the tree. The corresponding dynamic programming table is denoted by $\OPT(G_u)$. At the root we get obtain the maximum size of $S$ by choosing the maximum value over all accepting classes. 

\begin{definition}\label{de:opt}
    Let $\cP=(G,X)$ be a regular predicate over graphs and sets of vertices or edges, let $w$ be an integer, and let $\cC$ be the corresponding set of homomorphism classes over $w$-terminal recursive graphs. 
    We associate to each $w$-terminal recursive graph $G$ an \emph{optimization table} $\OPT(G)$ of $|\cC|$ entries, such that, for each $c \in C$,
    $\OPT(G)[c] = \max\{\w(X) \mid h(G,X)=c\}$.
    If no such set $X$ exists, we set $\OPT(G)[c] = -\infty$. 
\end{definition}

Observe (see also~\cite{BoPaTo92,FraigniaudMRT24}) that, if $G=(W,W,E)$ is a base graph, then
\begin{equation}\label{eq:optbase}
   \OPT(G)[c]  = \w(\selected(c,W)) 
\end{equation}
for each $c \in \cC$ (or $-\infty$ if no such $W$ exists).
%
%
%
%
%
If $G=f(G_1,G_2)$, then $\OPT(G)$ can be computed based on $\OPT(G_1), \OPT(G_2)$, and~$f$. 
Indeed, for each $c \in \cC$, we have 
\begin{equation}\label{eq:optcomp}
    \begin{split}  
    \OPT(G)[c]  = \max_{c_1,c_2 \in \cC \text{~s.t.~} c = \odot_f (c_1,c_2)}  \OPT(G_1)[c_1] &+ \OPT(G_2)[c_2] - \\
      & -\w(\selected(c_1,W_1) \cap \selected(c_2,W_2))
    \end{split}
\end{equation}
As previously, the maximum of empty sets is considered to be $-\infty$.
Similarly to Lemma~\ref{le:decision}, we have:

\begin{lemma}[bottom-up optimization]\label{le:optim}
Let $\cP(G,X)$ be a regular predicate on graphs, corresponding to a formula $\varphi(X)$ with a free  vertex-set or edge-set variable. 
Let $G$ be a graph, and let $T=(I,F)$ be a tree decomposition of $G$ with bags $\{B_u \mid u \in I\}$.
Let $u$ be a node of the tree decomposition, with children $v_1,\dots, v_q$, $q\geq 0$).
$\OPT(G_u)$ can be computed using only $\Gbase_u$, and the values $B_{v_i}$ and $\OPT(G_{v_i})$ for all $i\in\{1,\dots,q\}$.
Moreover, for each $c \in \cC$, one can compute the $q$-uple of homomorphism classes $\ARGOPT_u[c] = (c_1,\dots c_q)$ such that the optimal partial solution $X_u$ of $G_u$ satisfying $h(G_u,X_u)=c$ was obtained by gluing optimal partial solutions $X_{v_i}$ of $G_{v_i}$ satisfying $h(G_{v_i},X_i)=c_i$, for all $i\in\{1,\dots,q\}$.
\end{lemma}

\begin{proof}
The proof is very similar to the one of Lemma~\ref{le:decision}. Again, $\OPT(\Gbase_u)$ can be computed by brute force (Eq.~\eqref{eq:optbase}), which also settles the case when $u$ is a leaf.
If $u$ is not a leaf, then we first compute $\OPT(\Geq{i}_u)$ for each $1\leq i \leq q$. Then we use the fact that, thanks to Equation~\ref{eq:Geq},  $ \Geq{i}_u = f_{(B_{v_i},B_u)}(G_{v_i},\Gbase_u)$, and the computation of $\OPT(\Geq{i}_u)$ is performed through Eq.~\eqref{eq:optcomp}.
For $i\geq 2$, we compute $\OPT(\Gleq{i}_u)$ from $\OPT(\Gleq{i-1}_u)$ and $\OPT(\Geq{i}_u)$, based on Equation~\ref{eq:optcomp}, and the fact that $\Gleq{i}_u = f_{(B_u,B_u)} (\Gleq{i-1}_u,\Geq{i}_u)$ (cf. Eq.~\eqref{eq:Gleq}).
Eventually, $\OPT(G_u) = \OPT(\Gleq{q})_u$.
Note that, at every application of Eq.~\eqref{eq:optcomp}, one can memorise, for each homomorphism class $c$ of the glued graph, the classes of the two subgraphs that produced the maximum weight. In particular, one can deduce the classes $(c_1,\dots,c_q) = \ARGOPT_u[c]$ for each class $c \in C$.
\end{proof}

\begin{algorithm}[htb!]
\caption{Sequential decision/optimization for regular $\cP$ on $G$}
\label{algo:sequential}
\begin{algorithmic}[1]
\Require tree decomposition $T=(V_E,F_T)$ of width $w-1$ of $G$ ; formula $\varphi$ and the corresponding homomorphism classes $\cC$, homomorphism function $h$ on base graphs, composition functions $\odot_f$, function $\selected$\Comment{See Theorem~\ref{th:reg}}

\bigskip

\State \textbf{Bottom-up phase} on tree $T$, computes $h(G_u)$/$\OPT(G_u),\ARGOPT_u$ for each node $u$
\For {each node $u$ of $T$ from the bottom to the root}
    \State let $v_1,\dots, v_q$ be the children of $u$ in $T$   \Comment{These nodes have been treated before $u$}
    \State // Decision problems:
    \State Compute $h(G_u)$ from $\Gbase_u$, $B_{v_i}$ and $h(G_{v_i})$, $1\leq i \leq q$ using Lemma~\ref{le:decision} \label{i:decision}
    \State // Optimization problems:
    \State Compute $\OPT(G_u)$ and $\ARGOPT_u$ from $\Gbase_u$, $B_{v_i}$ and $\OPT(G_{v_i})$, $1\leq i \leq q$ by Lemma~\ref{le:optim} \label{i:optim}    

\EndFor

\bigskip
\State \textbf{Decision at the root $r$ (decision problems only).} Return $\textrm{true}$ if $h(G_r)$ is an accepting class, otherwise return $\textrm{false}$.

\bigskip

\State \textbf{Top-down phase (optimization only)}\label{i:top-down}
\For {each node $u$ of $T$ from the top to the bottom}
    \If{$u$ is the root}
        \State choose class $c_u\in C$ such that $c_u$ is accepting and $\OPT(G_u)[c]$ is maximized
        \State \textbf{if} no such class exists \textbf{then} reject \textbf{end if}
    \Else
        \State $c_u$ is the class received by $u$ from its parent in $T$
    \EndIf
    \State // Selects edges/vertices of the optimal solution:
    \State Mark all edges/vertices of $\selected(c_u,B_u)$ as “selected" 
    \If{$u$ is not a leaf}
        \State let $v_1,\dots, v_q$ be the children of $u$ in $T$ 
        \State let $(c_1,\dots, c_q) = \ARGOPT_u[c_u]$ \Comment{$c_i$ is the optimal class for child $v_i$, see Lemma~\ref{le:optim}}
        \State Send to each $v_i$ value $c_i$, for all $i=1,\dots,q$
    \EndIf
\EndFor\label{i:end}
\end{algorithmic}
\end{algorithm}

\subsection{The Sequential Algorithm}

Algorithm~\ref{algo:sequential}  simultaneously presents the model-checking of regular predicates on graphs, and the optimization protocol for predicates on graphs and sets. 

For the decision problem, we simply compute bottom-up, for each node $u$, the class of $h(G_u)$ using Lemma~\ref{le:decision}. At the root $r$, the algorithm accepts if $h(G_r)$ is an accepting class.

For optimization problems we need to construct bottom-up the full optimization tables $\OPT(G_u)$. At the root, it suffices to find an accepting class $c_r$ maximizing $\OPT(G_u)$. Then the value $\OPT(G_u)[c_r]$ is the maximum size of set $X$ satisfying $\cP(G,X)$. In order to retrieve the optimal set $X_{\OPT}$ itself, we ``roll back'' the whole dynamic programming process, in  a top-down phase.  Specifically, at the root node $u=r$, we choose the class $c_u$ with maximum value of $\OPT(G_u)$ over all accepting classes. In particular, $\selected(c,B_u)$ indicates the vertices, or edges, of $G[B_u]$ contained in the optimal solution. Moreover, by Lemma~\ref{le:optim}, $\ARGOPT_u[c]$ provides the tuple of classes $(c_1,\dots,c_q)$ from which the optimal class $c_u$ has been obtained through gluing of partial solutions of children of~$u$. Therefore $u$ can ``inform'' the $i$th child that its optimal class is~$c_i$. Each node performs the same procedure by decreasing depth, starting from the class received from its parent.

\section{Distributed construction of the elimination tree}

Our $\CONGEST$ protocol constructing an elimination tree of depth smaller than $2^d$ for graphs of treedepth at most $d$ is depicted in Algorithm~\ref{algo:congestElimT}. A similar  question was previously addressed in~\cite{NesetrilM16} (see Section~\ref{sec:conclusion}).

\begin{lemma}\label{lem:congestElimTree}
Let $G=(V,E)$ be the (connected) input network, and let $d\geq 1$ be an integer. There exists  an algorithm  performing in $O(2^{2d})$ rounds in \CONGEST\/ that outputs either an elimination tree $T=(T,F)$ of $G$ with  depth at most~$2^d$, or reports that $\td(G)>d$.
In the former case, each node $u\in V$ knows its parent and its children in the tree~$T$ at the end of the algorithm, as well as the depth of~$T$.
\end{lemma}

\begin{proof}
    Algorithm~\ref{algo:congestElimT} constructs an elimination tree following the same approach as Lemma~\ref{le:tdAlter}, in a greedy manner. Since $G$ is connected, it starts with a root vertex $v=r$ (chosen arbitrarily), and then constructs elimination trees of $G\smallsetminus v$, by treating each component separately. The components of $G\smallsetminus v$ are identified by their leader with the smallest node's identifier of the component. Each unmarked node eventually knows its leader (Instruction~\ref{i:leader}). For a component with leader~$\ell$, we choose as root of the component a vertex that is adjacent to $v$  (Instruction~\ref{i:rootell}). In particular, every edge of the tree is also an edge of $G$  (see Instruction~\ref{i:parent}).

    The construction preserves the following invariant. The tree constructed after step~$i$ is an elimination tree of the subgraph induced by the marked vertices. Moreover, for each connected component of unmarked vertices, its outgoing edges are solely incident to a path from the root and a vertex $v$ of depth~$i$. In particular, at the end, $T$ is an elimination tree of~$G$. 
Furthermore, $T$ is a subtree of $G$. Therefore, by Lemma~\ref{le:greedyDFS}, if $\td(G)\leq d$ then the depths of $T$ is smaller than $2^d$ as requested, and the algorithm marks all vertices in less than $2^d$ phases. Consequently, if some vertices remain unmarked after this many rounds (Instruction~\ref{i:largetd}), we correctly assert that $\td(G)>d$.

    Regarding the round-complexity, observe that, at each step, there is a call to Algorithm $\leader$ on the set of unmarked nodes (see, e.g.,~\cite{HiSu20} for a detailed description of a leader-election algorithm). Its round complexity is $O(\diam(G))$. The diameter of $G$ is $O(2^{\td(G)})$, and thus is it at most~$O(2^d)$ (we can adapt algorithm $\leader$ such that, if it does not succeed in $O(2^d)$ rounds, then it rejects, which is correct as, in this case, $\td(G)>d$).
\end{proof}

\begin{algorithm}[htb!]
\caption{$\CONGEST$ algorithm computing an elimination tree of $G$ in $O( 2^{2d})$ rounds}
\label{algo:congestElimT}
\begin{algorithmic}[1]
\Require Protocol $\leader(G=(V,E),U)$ with a set $U \subseteq V$ of distinguished vertices (i.e., each vertex knows whether it belongs to $U$). 
After $O(\diam(G))$ 
$\CONGEST$ rounds, each vertex $u$ will know $\leader(u)$, the minimum identifier in the component of $G[U]$ containing $u$. 

\State Apply $\leader$ on all vertices of $G$
\State Let $r$ be the unique node such that $r=\leader(r)$\Comment{Unique since $G$ is connected}
\State Set $\parent(r) =r$ \Comment{$r$ is the root of the tree}
\State Mark vertex $r$ \Comment{Marked vertices are those already placed in the tree}
\State Set $\depth(r)=1$
\For {step $i=2$ to $D=2^d-1$}
    \State // At step $i$ we identify the nodes of $T$ of depth $i$ 
    \State Apply $\leader$ on all unmarked vertices of $G$ \Comment{$O(2^d)$ rounds}\label{i:leader}
    \State Each unmarked vertex $u$ broadcasts $(\leader(u),u)$ to its neighbours \Comment{One round}
    \For{each marked vertex $v$ of depth $i$} \Comment{All in a same round}
        \For {each $\ell$ among values $\leader(u)$ received by $v$}
            \State $v$ picks the corresponding $u(\ell)$ of minimum $\id$ \Comment{$u(\ell)$ is a new node of depth $i$}\label{i:rootell}
            \State $v$ adds $u(\ell)$ to the list of its children
            \State $v$ sends to $u(\ell)$ a message with its $\id$ indicating that it becomes its parent.\label{i:parent} 
       \EndFor
    \EndFor
    \For{each vertex $u$ that receives such a message from some $v$} \Comment{All in a same round}
        \State $u$ sets $\parent(u)=v$, $\depth(u)=i$ and marks itself 
    \EndFor
\EndFor
\If{some vertex $u$ is still unmarked}\label{i:largetd}
    \State $u$ rejects because $\td(G)>d$ \Comment{$G$ contains a path with more than $D$ vertices}
\EndIf
\end{algorithmic}
\end{algorithm}

\begin{lemma}\label{lem:congestCanonTW}
Let $G=(V,E)$ be the input network, and let us assume that an elimination tree $T=(F,V)$ of depth smaller than $2^d$ has been constructed as in Lemma~\ref{lem:congestElimTree}.
There is a $\CONGEST$ algorithm constructing the canonical tree decomposition $(T,(B_u)_{u \in V})$ in $O(2^d)$ rounds. At the end of the algorithm, each node $u$  knows its bag $B_u$ as well as the graph $G[B_u]$ induced by the bag.
\end{lemma}

\begin{proof}
The algorithm proceeds top-down. For each round $i=1,\dots, D=2^d-1$, every node $u$ at depth $i$ computes $B_u$ and $G[B_u]$. Observe that when $u$ is the root, $B_u$ is a singleton so $G[B_u]$ is trivial. If $u$ is not the root, then $u$ has received $B_v$ and $G[B_v]$ from its parent $v$. Observe that $B_u = B_v \cup \{u\}$ and the edges of $G[B_u]$ are the edges of $G[B_v]$, plus the edges incident to $u$. Therefore, node $u$ is able to compute the information from its parent, and to transmit it to its children. 
\end{proof}

\section{Distributed model checking and optimization}

We have now all ingredients to prove our main result. 

\begin{theorem}[Distributed decision and optimization]\label{th:main} ~
\begin{itemize}
\item For any closed MSO formula $\varphi$, there exists an algorithm which, for any $n$-node graph~$G$, and any $d\geq 0$, decides whether $G \models \varphi$, or reports ``large treedepth'' if $\td(G)>d$, running in $O(2^{2d})$ rounds in the  $\CONGEST$ model.
\item For any MSO formula $\varphi(S)$ with a free  variable $S$ representing a vertex-set, or an edge-set, there exists an algorithm which, for any $n$-node graph~$G$, and any $d\geq 0$, selects a set  $S$  of maximum weight satisfying $G \models \varphi(S)$, or reports ``large treedepth'' if $\td(G)>d$, running in $g(d,\varphi)$ rounds in the  $\CONGEST$ model for some function~$g$.
\end{itemize}    
\end{theorem} 

\begin{proof}
By Lemmas~\ref{lem:congestElimTree} and~\ref{lem:congestCanonTW},  one can construct a canonical tree decomposition $T=(V,F)$  of $G=(V,E)$, with bags $\{B_u \mid u\in V\}$ of width at most~$2^d$ (or correctly reject because $\td(G)>d$), in $O(2^{2d})$ rounds. Moreover each node $u$ knows its parent $\parent(u)$, its bag $B_u$, the graph $G[B_u]$, and its depth in~$T$. By construction, the tree $T$ is a subgraph of $G$. It remains to show that, based on these elements, one can implement the sequential algorithm (cf. Algorithm~\ref{algo:sequential}) in $\CONGEST$.

Let us first consider model-checking of a closed formula $\varphi$. We describe how the bottom-up phase, and the decision at the root in Algorithm~\ref{algo:sequential} can be implemented in $\depth(T)$ rounds. Let us consider all steps $j\in\{1,\dots,\depth(T)\}$, where each step consists of a single round. At step~$j$, all nodes $u$ of depth $k=\depth(T)-j+1$ can compute the homomorphism classes $h(G_u)$ in parallel ($k$ decreases from $\depth(T)$ to~1), and can send the results of this computation to their parents. Indeed, if $u$ of depth $k$ is a leaf, then it  has all information needed to compute $h(G_u)$ already, because it only needs to know graph $\Gbase_u = G[B_u]$ (see Instruction~\ref{i:decision} of Algorithm~\ref{algo:sequential}, and Lemma~\ref{le:decision}). 
If $u$ it is not a leaf, then it also needs the bags $B_{v_i}$, and the homomorphism classes $h(G_{v_i})$ from all its children~$v_i$, $\{1,\dots,q\}$. But, at step $j$, node $u$ has precisely already received these information from its children, who have sent them at the previous step $j-1$.
The decision at the root can be performed at round~$1$. The root accepts or rejects depending on its homoporphism class, as in Algorithm~\ref{algo:sequential}, and all other nodes accept. Therefore, if $G \models S$, then all nodes accept, otherwise the root rejects.
Note that each message consists of a homomorphism class, thus the size of the messages is a constant. More precisely, messages are of size  $\log |\cC|$ bits.

Next, we consider optimization problems $\varphi(S)$.
In the bottom-up phase of Algorithm~\ref{algo:sequential}, the main difference with the model-checking case is that each step~$j$ now requires $|\cC|$ rounds to be performed. Again, nodes of depth $k = \depth(T)-j+1$ can perform their computations in parallel (cf. Instruction~\ref{i:optim} of Algorithm~\ref{algo:sequential}, and see Lemma~\ref{le:optim}). However,  they need to broadcast the whole table $OPT(G_u)$, which contains $|\cC|$ entries of size $O(\log n)$ because each entry corresponds to the weight of an edge or vertex subset (recall that the weights are supposed to be polynomial in~$n$). Therefore, when a node $u$ aims at performing its computation, it has indeed received all necessary inputs from its children.
The top-down phase (Instructions~\ref{i:top-down} to~\ref{i:end}) only requires $\depth(T)$ rounds. At phase~$j$, all nodes at depth $j$ work in parallel. From its own optimal class $c_u$, every node $u$ can retrieve the corresponding optimal classes $c_{v_i}$, for each of its children $v_i$, $i\in\{1,\dots,q\}$. Thus, in a single communication round, $u$ can send the class $c_{v_i}$ to each child $v_i$. Also, at round $j=\depth(u)$, node $u$ computes $\selected(c_u,B_u)$. If the set $S$ that we are aiming at optimizing is a vertex set. If $u \in \selected(c_u,B_u)$, then $u$ must be in the optimal solution $S_{\OPT}$, and it thus selects itself. Note that other notes of bag $B_u$ might be in~$S_{\OPT}$. However,  since they are ancestors of $u$ in the tree $T$, they have selected themselves at some previous step. If the set $S$ is an edge set, then $u$ just selects the edges incident to it that belong to $\selected(c_u,B_u)$. 

The correctness of the distributed protocol, and hence the proof of Theorem~\ref{th:main}, follows directly from the correctness of Algorithm~\ref{algo:sequential}.
\end{proof}

Let us complete the section by extensions to labeled graphs, and to counting problems.

\paragraph*{Labeled graphs.} 

Definition~\ref{de:reg}, and Theorem~\ref{th:reg} extend to formulas on \emph{labeled} graphs, i.e., one can add to the input graph a constant number of labels $L_1,\dots,L_{\ell}$ on vertices, and  on edges. Labels are expressed as unary predicates that can be used in the MSO formulas, in addition to the binary predicates $\adj(x,y)$ and $\inc(x,e)$. So, as incidence and adjacency, labels are part of the input.
For example, one can use vertex labels $\red$ and $\blue$, and ask for the minimum set of blue vertices that dominates all red vertices. This corresponds to solving the problem $max\varphi(S)$ for the following formula, with all weights set to~$-1$: 
\[
\varphi(S) = \big(\forall x \in S \  \blue(x) \big) \land \Big(\forall y \; 
\neg \big ( \red(y)\land (\forall x \in S  \; \neg\adj(x,y) )\big)\Big).
\]
Theorem~\ref{th:main} on distributed decision and optimization also applies to labeled graphs. 

A slightly different application of labeled graphs is problem $optmarked\varphi$ in which we are given a set of vertices (or edges), marked via a unary predicate~$\Mark$, and a formula~$\varphi(S)$, and the question is whether the marked set corresponds to a maximum weight set satisfying~$\varphi$. In this setting, we can express problems such as: Is the marked set a minimum feedback vertex set? Is the marked set a minimum-weight spanning tree? For answering such questions in the \CONGEST\/ model, it is sufficient to  modify the bottom-up phase of our algorithm such that the root $r$ of the decomposition tree collects:
\begin{enumerate}
    \item the optimization table $\OPT(G_r)$ for formula $\varphi(S)$, by executing the optimization protocol;
    \item the decision for a closed formula $\psi$ obtained from $\varphi$ by replacing $S$ by $\Mark$;
    \item the total weight of the marked set, obtained by summing up, at every node $u$, the weight of the marked vertices of $G[u]$, which are obtained from the weights at the children nodes.
\end{enumerate}
Eventually, the root accepts if (1)~$\psi$ is accepted (indicating that the marked set of nodes or edges satisfies the formula), and (2)~the weight of the marked set is equal to the maximum of $\OPT(G_r)$ (confirming that the set is indeed an optimal one). All other vertices accept.

Therefore, problem $optmarked\varphi$ can also be solved in $g(d,\varphi)$ rounds in \CONGEST. 

\paragraph*{Counting.} 

The results of Borie, Parker and Tovey~\cite{BoPaTo92} actually concern formulas with an arbitrary number of free variables, each variable being of type ``vertex'' or ``edge'', or ``vertex set'' or ``edge set''. For example, by considering three vertex variables $x_1,x_2,x_3$, one can easily express a formula $\varphi(x_1,x_2,x_3)$ stating that $x_1,x_2$ and $x_3$ form a triangle. 
%
%

Definition~\ref{de:reg} and Theorem~\ref{th:reg} extend to predicates $\cP(G_1,X_{11},\dots,X_{1p})$, and to formulae $\varphi(X_{11},\dots,X_{1p})$ with an arbitrary number of variables, where each variable $X_i$ denotes a vertex or an edge, or a vertex subset or an edge subset of $G$. For each possible assignment of variables with corresponding values, $\cP$ is either true or false.
Gluing functions $f$ naturally extend to tuples $(G_1,X_{11},\dots,X_{1p})$ and 
$(G_2,X_{21},\dots,X_{2p})$, the operation being valid only under some specific conditions. We refer again to~\cite{BoPaTo92} for a full description of the gluing, and of the interpretation of the values of the variables.

In Definition~\ref{de:reg}, the two conditions for regularity become:
\begin{enumerate}
\item If $h(G_1,X_{11},\dots,X_{1p}) = h(G_2,X_{21},\dots,X_{2p})$ then 
$$\cP(G_1,X_{11},\dots,X_{1p}) = \cP(G_2,X_{21},\dots,X_{2p});$$
\item For any two $w$-terminal recursive graphs $G_1$ and $G_2$ and variables $X_{11},\dots,X_{1p}$ of $G_1$, and $X_{21},\dots,X_{2p}$ of $G_2$,
\begin{equation*}
\begin{split}
 h\Big(\circ_f\big((G_1,X_{11},\dots,X_{1p}),(G_2,X_{21},\dots,X_{2p})\big)\Big) = \\ =\odot_f\big(h(G_1,X_{11},\dots,X_{1p}),h(G_2,X_{21},\dots,X_{2p})\big).    
\end{split}
\end{equation*}
\end{enumerate}

Borie, Parker and Tovey~\cite{BoPaTo92}, propose a sequential algorithm for  problem $count\varphi$, counting the number of different true assignments of a formula $\varphi(X_1,\dots,X_p)$, on $w$-terminal recursive graphs. We can restate their technique in our framework, as follows.  Let  $\mathsf{COUNT}(G)$ be the table counting, for each homomorphism class $c \in \cC$, the number of different partial assignments to variables $X_1, \dots, X_p$ such that $h(G,X_1,\dots,X_p) = c$.
This table can be computed in constant time on base graphs. Similarly to Equation~\ref{eq:optcomp} for optimization problems, \cite{BoPaTo92} proposes an approach to compute, for graph $G=f(G_1,G_2)$, the table $\mathsf{COUNT}(G)$, using only function~$f$, and tables  $\mathsf{COUNT}(G_1)$ and $\mathsf{COUNT}(G_2)$. Similarly to Lemma~\ref{le:optim}, we can derive a lemma for counting problems, describing, at each node~$u$, the computation of table $\mathsf{COUNT}(G_u)$ from tables $\mathsf{COUNT}(G_{v_i})$ of children nodes $v_i$, together with their bags $B_{v_i}$ and the base graph $\Gbase_u$. 

Therefore, problem $count\varphi$ can be solved in $g(d,\varphi)=O(1)$ rounds on graphs of treedepth at most~$d$, which is the same round-complexity as for optimization. In particular, triangle counting can be performed in a constant number of rounds in \CONGEST, on bounded treedepth graphs. Note that for the triangle counting problem, the design of the dynamic programming tables, and of the update functions is a simple but enriching exercise.

\section{Applications to $H$-freeness for graphs of bounded expansion}

For the many alternative definitions of graphs of bounded expansion, we refer to the book of Ne\v{s}et\v{r}il and Ossona de Mendez in \cite{NesetrilOdM12}. In terms of applications, we simply recall that the class of planar graphs, and, more generally, every class of graphs excluding a fixed minor, are classes of graphs of bounded expansion. It is known that graphs of bounded expansion admit so-called \emph{low treedepth decompositions}.

\begin{theorem}[\cite{NesetrilOdM12}]\label{theo:low-treedepth-decomposition}
Let $\mathcal{G}$ be a class of graphs of bounded expansion. There is a function $f:\mathbb{N} \to \mathbb{N}$ such that, for every  integer $p>0$, and every graph $G = (V, E)\in \mathcal{G}$, the vertex set of $G$ can be partitioned into at most $f(p)$ parts $V_1, \dots, V_{f(p)}$, such that the union of any $q$ parts, $1\leq q \leq p$ parts induces a subgraph of $G$ with treedepth at most~$q$.
\end{theorem}

A partition satisfying the property of Theorem~\ref{theo:low-treedepth-decomposition} is called a \emph{low treedepth decomposition} of $G$ for parameter $p$. Interestingly, low treedepth decompositions can be efficiently computed in \CONGEST, i.e., each vertex can compute the index $i\in \{1,\dots,f(p)\}$ of the part to which it belongs.

\begin{theorem}[\cite{NesetrilM16}]\label{th:congestltd}
   For every graph class $\mathcal{G}$ with bounded expansion, and every positive integer~$p$, a low treedepth decomposition of $G$ for parameter $p$ can be computed in $O(\log n)$ rounds in \CONGEST. 
\end{theorem}

The constant hidden in the big-O notation in the statement of Theorem~\ref{th:congestltd} depends on the class~$\mathcal{G}$ and on the parameter~$p$, and it is quite huge. The proof of Theorem~\ref{th:congestltd} is sophisticated, but the algorithm is actually quite simple. It is merely based on the fact that  graphs with bounded expansion have bounded degeneracy, and on the use of standard distributed tools for approximating the degeneracy of a graph in \CONGEST. Combining Theorem~\ref{th:main} with Theorem~\ref{th:congestltd}, we can establish the following. 

\begin{corollary}\label{co:main}
    Let $\mathcal{G}$ be a class  of graphs with bounded expansion, and let $H$ be a connected graph. Deciding $H$-freeness for graphs in $\mathcal{G}$ can be achieved $O(\log n)$ rounds under the \CONGEST\/ model.    
\end{corollary}

\begin{proof}
The algorithm works as follows. Let $p$ be the number of vertices of $H$. First, compute a low treedepth decomposition $V_1,\dots,V_{f(p)}$ of the input graph $G=(V,E)$ into $f(p)$ parts for parameter $p$ using Theorem~\ref{th:congestltd}. Then, for every non-empty set $I\subseteq [f(p)]$ with $|I|\leq p$, let $G_I = G[\cup_{i \in I} V_i]$ be the graph induced by the parts $V_i$, $i\in I$. Note that there are at most ${f(p)}\choose{p}$ such subsets $I$, that is, a constant number of choices for~$I$. Also observe that if a copy of $H$ exists in graph $G$, then this copy of $H$ belongs to one of the graph $G_I$, where $I$ is the set of the parts that are containing at east one vertex of the copy. It is therefore sufficient to run the algorithm in Theorem~\ref{th:main} on each graph $G_I$ in parallel, and to reject if one of the parallel executions finds a copy of $H$. This is doable because (1)~$G_I$ is of treedepth at most $p$, (2)~if a copy of $H$ exists, then it will be found in a connected component of $G_I$, thanks to the fact that $H$ is connected, and (3)~the property ``$G_I$ is $H$-free'' can be expressed as an MSO formula (actually, even as an FO formula), with $p$ variables, one for each vertex of $H$. For instance, for a graph $H=(V_H,E_H)$ with $V_H=\{1,2,\dots,p\}$, we can use the formula 
$$
\varphi_H = \neg \exists x_1,x_2,\dots,x_p \left(\bigwedge_{\{i,j\}\in E_H} \adj(x_i,x_j)\right) \land \left( \bigwedge_{\{i,j\}\notin E_H} \neg \adj(x_i,x_j)\right).
$$
Consequently, the algorithm rejects if and only if the input graph $G$ contains a copy of $H$, as desired.
\end{proof}

In particular, Corollary~\ref{co:main} proves that $H$-freeness can be solved in $O(\log n)$ rounds in planar graphs under \CONGEST. In contrast, for arbitrary graphs, even $C_4$-freeness requires $\Omega(\sqrt{n})$ rounds, and, for every $p \geq 2$, there are $O(p)$-vertex graphs $H$ for which $H$-freeness requires $\Omega(n^{2-1/p})$ rounds~\cite{FischerGKO18}. Note that $H$-freeness as considered in Corollary~\ref{co:main} can be considered in the usual sense (i.e., the input graph does not contain any copy of $H$ as an induced subgraph), but also in the mere sense that there are no copies of $H$ as a (non necessarily induced) subgraphs, by a straightforward adaptation of the MSO formula describing the problem.

\section{Conclusion}\label{sec:conclusion}

In this paper, we established a meta-theorem about MSO formulas on graphs with bounded treedepth within the \textsf{CONGEST} model. 
Treedepth plays a fundamental role in the theory of sparse graphs of Ne\v{s}et\v{r}il and Ossona de Mendez~\cite{NesetrilOdM12}. In particular, decomposing a graph in graphs of bounded treedepth is the crucial step in deriving a linear-time model-checking algorithm for FO on graphs of bounded \emph{expansion} in the sequential computational model. Graphs of bounded expansion contain bounded-degree graphs, planar graphs, graphs of bounded genus,  graphs of bounded treewidth, graphs that exclude a fixed minor, etc.
Model-checking for FO on graphs of bounded expansion cannot be achieved in the \CONGEST\/ model since, as we already mention,  checking an FO predicate as simple as ``there is at least one vertex of degree~$>2$'' requires $\Omega(n)$ rounds in $n$-node trees. 
Nevertheless, there might exist some fragments of FO that could be tractable on graphs of bounded expansion in the distributed setting. It would be interesting to identify the exact boundaries of intractability in this context, regarding both distributed decision, and distributed certification. An initial step in this direction was taken by Ne\v{s}et\v{r}il and Ossona de Mendez in \cite{NesetrilM16}, resulting in a distributed algorithm for computing a low treedepth decomposition of graphs of bounded expansion, running in $O(\log n)$ rounds under \CONGEST. As we illustrated, this results allows to efficiently decide   FO-expressible decision problems (such as $H$-freeness, for $H$ connected) for classes of graphs with bounded expansion, in $O(\log n)$ rounds. We restate the open question of~\cite{NesetrilM16}: Given a \emph{local} FO formula $\varphi(x)$, i.e., a formula where $\varphi(x)$ depends on a fixed-radius neighborhood of vertex $x$ only, can we mark all vertices satisfying~$\varphi$ in $O(\log n)$ rounds?


\bibliography{biblio-FO-treedepth}

\end{document}